\documentclass[12pt, reqno]{amsart}
\usepackage[thm]{macros}
\usepackage{chngcntr} %

\usepackage[all=normal, paragraphs=tight, bibbreaks=tight, floats=tight,bibnotes=tight]{savetrees}

\newcommand{\bhs}{\mathfrak{s}}

\newcommand{\cW}{\mathcal{W}}

\crefname{ex}{Example}{Examples}

\newgeometry{margin=1.25in}

\title{Reinterpreting demand estimation}
\author{Jiafeng Chen \\ Department of Economics, Stanford University \\
jiafeng@stanford.edu}
\date{\today. This paper is largely motivated by  discussions and issues that arose during
 a thought-provoking seminar by Steve Berry  at Stanford GSB in March 2025, to whom I am
 grateful. I thank Isaiah Andrews, Dmitry Arkhangelsky, Lanier Benkard, Kirill Borusyak,
 Jos\'e Ignacio Cuesta, Liran Einav, Matt Gentzkow, Jeff Gortmaker, Guido Imbens, Ravi
 Jagadeesan, Ivana Komunjer, Sokbae Lee, Lihua Lei, Ariel Pakes,    Ashesh Rambachan,
 Steve Redding, Brad Ross, Jonathan Roth, Bernard Salani\'e, Jesse Shapiro, Paulo
 Somaini, Shoshana Vasserman, and participants of the 2025 Bay Area IO Fest  for helpful
 discussions. }

\begin{document}

\maketitle

\begin{abstract}
This paper clarifies how and why structural demand models 
\citep{berry2014identification,berry2024nonparametric} predict unit-level counterfactual
 outcomes. We do so by  casting structural assumptions equivalently as restrictions on
 the joint distribution of potential outcomes. Our reformulation highlights a 
 \emph{counterfactual homogeneity} assumption underlying structural demand models: The
  relationship between counterfactual outcomes is assumed to be identical across markets.
  This assumption is strong, but cannot be relaxed without sacrificing identification of
  market-level counterfactuals. Absent this assumption, we can interpret model-based
  predictions as extrapolations from certain causally identified average treatment
  effects. This reinterpretation provides a conceptual bridge between structural modeling
  and causal inference.
\end{abstract}

\pagestyle{plain} 

\newpage

\section{Introduction}

Predicting counterfactual outcomes for individual units is central to many areas of
economics. In industrial organization, for instance, prices set by firms depend on market
shares at counterfactual prices; thus, predictions of these counterfactuals  yield
markups and marginal costs. Structural econometric methods are often motivated by their
ability to predict counterfactual outcomes for individual units. Once researchers fit a
model to observed data, the model implies counterfactual outcomes for all units.  By
contrast, the literature on causal inference
\citep{splawa1990application,rubin1974estimating} generally focuses on recovering \emph
 {average} counterfactuals. For unit-level counterfactuals, causal inference methods are
 typically informal---e.g., extrapolating from average treatment effects (ATEs) among
 observably similar units---if they are produced at all.

These ``two cultures'' for predicting counterfactuals, to quote
\citet{breiman2001statistical}, face parallel critiques. The causal inference literature
shows that certain average counterfactuals are identified through credible treatment
variation. However, these averages often are  not themselves of economic interest.
Extrapolating them to individuals would only be valid under constant treatment effects,
which severely restricts unobserved heterogeneity. In contrast, structural methods impose
modeling assumptions up front to directly target unit-level counterfactuals. But these
predictions seem to hinge on the model: It can be unclear how to interpret them without
the model.\footnote{As examples of these respective critiques in the literature,
\citet{berry2021foundations} write of the treatment effects literature, ``In empirical
settings with endogeneity and multiple unobservables, economists often settle for
estimation of particular weighted average responses (e.g., a local average treatment
effect); but this is a compromise poorly suited to the economic questions that motivate
demand estimation, as these typically require the levels and slopes of demand at specific
points.'' \citet{nevo2010taking} argue that heterogeneity is sufficiently strong for
average effects over past mergers to not be informative: ``As our discussion of merger
analysis illustrates, industrial organization economists seem far more concerned than
labor economists that environmental changes are heterogeneous, so that useful estimates of
average treatment effects in similar situations are not likely to be available.''
\citet{angrist2010credibility} write of structural industrial organization, ``In this
framework, it’s hard to see precisely which features of the data drive the ultimate
results.''}

To reconcile and bridge the two cultures, we ask: First, do structural models avoid
restricting unobserved heterogeneity, or do they too extrapolate from averages? Second,
how should we interpret structural model predictions when the model is only an
approximation?  This paper studies these questions in the context of canonical structural
demand models, in both settings with market-level shares \citep
{berry1995automobile,berry2014identification} and settings with demographic-specific
market shares \citep{berry2024nonparametric}.

In either case, we cast the structural model equivalently as restrictions on the joint
distribution of potential outcomes.  This equivalent reformulation---in the spirit of
\citet{vytlacil2002independence,vytlacil2006ordered}---does not simply declare that the
potential outcomes are generated from the corresponding structural model. Rather, the
model is reinterpreted as restricting the joint distribution of potential
outcomes.\footnote{\citet {vytlacil2002independence} shows that the selection model $D_i
 =
\one (\alpha + \beta z > \xi_i)$, for $\beta>0$, is equivalent to  the monotonicity
 restriction $\P (D_i(1) \ge D_i(0)) = 1$ \citep{imbensangrist}. The former is stated as
 a generative structural model of an endogenous treatment $D$, whereas the latter is
 stated as a restriction on the joint distribution of $ (D_i(1), D_i(0))$. }

We find that a key restriction that structural demand models impose is what we term \emph
{counterfactual homogeneity}: If the structural model holds, then two counterfactual
outcomes are deterministically related to each other through a function that is identical
across markets. Concretely, let $Y_i (a)$ denote the counterfactual market shares of a
given market $i$ under a bundle $a$ of product characteristics and prices. Counterfactual
homogeneity restricts that $\var(Y_i(a) \mid Y_i(a')) = 0$ for any bundle $a' \neq a$,
where the variance is taken over draws of markets over its population. This is a strong
restriction on the joint distribution of  $(Y_i(a), Y_i(a'))$. In this sense, these
structural
demand models \emph{do} restrict unobserved heterogeneity.

This restriction is not a flaw of these particular structural models---if we want models
that point-identify unit-level counterfactuals. Restricting unobserved heterogeneity is
necessary for identifying unit-level counterfactuals. Thus, counterfactual homogeneity
cannot be relaxed, unless we give up point-identification as well.\footnote{Of course,
point-identification is convenient but not necessary to make effective use of data.
Partial identification strategies \citep {molinari2020microeconometrics} are popular in
the literature on entry games \citep {ciliberto2009market} and revealed preference
\citep{pakes2015moment}. It is also possible to partially relax counterfactual homogeneity
by demanding that only certain counterfactuals---e.g., counterfactual in prices---are
identified \citep{andrews2023structural,borusyak2025estimating,newpaper}. } Additional
functional form assumptions in \citet{berry2014identification,berry2024nonparametric},
which are sufficient but not necessary, ensure that the homogeneous relationship linking
counterfactuals is uniquely recovered by instrument variation. Modulo these additional
functional forms, nonparametric structural demand models are indeed minimally restrictive
for point-identified unit counterfactuals.

Nevertheless, just as we are uneasy with homogeneous treatment effects when extrapolating
from ATEs, counterfactual homogeneity should also give us pause. Counterfactual
homogeneity meaningfully restricts how markets may be different from each other.\footnote
{This differs from within-market consumer heterogeneity allowed by BLP \citep
{berry1995automobile}. In our notation, consumer heterogeneity corresponds to whether $a
\mapsto Y_i(a)$ is a flexible function. In contrast, counterfactual homogeneity are
 restrictions on how different the demand curves $Y_i(\cdot)$ and $Y_j(\cdot)$ for two
 markets can be.} It rules out, for instance, settings in which each market aggregates a
 population of consumers with heterogeneous preferences, but different markets have
 unobservably different populations of consumers.  It also imposes that markets with the
 same observed conditions necessarily have identical counterfactuals everywhere---ruling
 out demand surfaces that intersect nontrivially.

These implications of counterfactual homogeneity are demanding. This reflects
that the structural models are simplifications and are unlikely to hold literally. Thus,
unit-level counterfactuals under counterfactual homogeneity are better interpreted as
extrapolated predictions rather than as point-identified treatment effects \citep
{kline2019heckits}. We formalize how \citet{berry2014identification} extrapolate from
average effects. We also show that this kind of extrapolation is essentially what large
classes of structural models do. This exercise clarifies the value of structural models.
Many counterfactual predictions are effectively extrapolating from certain ATEs---acting
as if every unit has the same treatment effect. Structural models are additionally helpful
in motivating \emph{which} ATEs to extrapolate from.

This paper contributes to a literature that bridges causal inference and structural
modeling 
\citep{andrews2023structural,kline2019heckits,borusyak2025estimating,kong2024nonparametric,humphries2025conviction,torgovitsky2019nonparametric,mogstad2024instrumental,angrist2000interpretation,conlon2021empirical}.
 This paper is also related to transformation 
 models 
 and other simultaneous equation
 models 
 \citep
  {chiappori2015nonparametric,vuong2017counterfactual,benkard2006nonparametric,matzkin2008identification}.
  Counterfactual homogeneity is related to a literature on omitted parameter
  heterogeneity 
 \citep{chesher1984testing,hahn2014neglected,qian2025testing}.
 Like
 \citet
  {berry2014identification,berry2024nonparametric,vytlacil2002independence,kline2019heckits},
  this paper's primary focus is conceptual---the identification and expressivity of
  workhorse models.\footnote{The models estimated in practice are typically versions of 
  \citet{berry2014identification,berry2024nonparametric} with additional parametric
  assumptions. \citet{compiani2018nonparametric} studies nonparametric estimation for
  \citet{berry2014identification}. }

  This paper proceeds as follows. \Cref{sec:market_level_data} derives equivalent
  assumptions to \citet{berry2014identification}. \Cref{sec:discussion} discusses
  counterfactual homogeneity, derives its necessity, and derives an equivalence between
  structural model predictions and extrapolations from average treatment effects. \Cref
  {sec:demographics_specific_market_shares} derives equivalent assumptions to 
  \citet{berry2024nonparametric} and examines the extent to which models with micro-data
   allow for counterfactual heterogeneity.

\section{Market-level data}
\label{sec:market_level_data}

We start with a standard model of differentiated products \citep
{berry1995automobile,berry2014identification} in potential outcomes notation. Markets are
i.i.d. draws from a population $F^*$, following \citet
{freyberger2015asymptotic}. Each market contains the same $J \in \N$ inside options. The
observed data in each market take the form $(Y, A, Z)$. Here $Y \in
\mathcal Y \subset [0,1]^J$ is the vector of observed market shares,
$A \in
\mathcal A\subset \R^{J \times d_a}$ is the bundle of prices and characteristics
associated with each of the $J$
 goods, and $Z \in \mathcal Z \subset \R^{d_z}$ is a vector of external instruments that
 includes the exogenous entries of $A$. For concreteness, we may write $A =
 (A_1,\ldots,A_J)$ for each product, for $A_j = (P_j, X_j)$ the prices and
 characteristics of a product. We view $A$ as a treatment acting on $Y$. 

 To embed the setup in the potential outcomes framework, let the random variable $Y
 (a)$ denote the potential outcome for a given market, were the bundle set to some
 counterfactual value $a$. The observed market shares $Y$ are generated from underlying
 potential outcomes, $Y = Y(A)$. We condition on other observed market covariates and
 omit them from notation.

Structural demand models posit that counterfactuals $Y (a)$ are generated
through  \[ Y(a) = \bhs(a, \xi) \numberthis.
\label{eq:structure}
\]
For instance, \citet{berry1995automobile} posit that market shares aggregate
heterogeneous consumers with Gumbel idiosyncratic preferences and
heterogeneous valuations for attributes ($\beta \sim G$): \[
	Y_j(a) = \bhs_j(a, \xi) = \int 
    \frac{e^{a_j'\beta + \xi_j}}{1 + \sum_{k=1}^J e^{a_k'\beta + \xi_k}} d G(\beta) 
    \text{ for some distribution $G$.}
\]
Here, the map $\bhs$ is indexed by the random coefficient distribution $G$.

The causal inference and structural demand literatures differ in their typical workflow.
The former usually focuses on average effects like $\E[Y(a_1) - Y(a_0)], \frac{d}{da} \E
[Y(a)]$, or conditional-on-covariates versions thereof
\citep{angrist2000interpretation}. If $a_1$ represents a price increase in good $j$
 relative to $a_0$, these parameters measure the average response of market shares to
 this price increase across some (sub)population of markets. These averages are in turn
 identified through various comparisons that exploit variation in $A$ induced by the
 instruments. Care is taken on restricting how $A$ responds to instruments
 (e.g., monotonicity, \citet{imbensangrist}) to ensure that instrument-level comparisons
 recover proper comparisons over endogenous treatments. 

On the other hand,  the structural demand literature is concerned with \emph
{unit-level} counterfactuals and views average effects as insufficient for scientific and
policy objectives. These unit-level counterfactuals are $Y(a_1) - Y(a_0)$, representing a
particular market's response to changes in $a$. Typically, models impose restrictions
on \eqref{eq:structure}, such that the structural error $\xi =
\bhs^{-1}(A, Y)$ can be recovered from observed variables with knowledge of $\bhs$, which
is itself identified through instrument variation.\footnote {In the case of
\citet{berry1995automobile}, the random coefficient distribution $G$ is identified under
additional parametric assumptions, implying that $\bhs$ is.} The model then identifies
counterfactual outcomes through $Y(a) =
\bhs (a, \bhs^ {-1} (A, Y))$. Identification of $\bhs$ requires assumptions on instrument
 strength, but need no monotonicity-type restrictions on the selection of $A$.

A standard intuition in causal inference  is that unit-level counterfactuals---or even the
distribution of individual treatment effects---are not identified even with a randomized
experiment, absent assumptions like rank invariance
\citep{doksum1974empirical,heckman1997making}.\footnote{This is even termed the 
``fundamental
 problem of causal inference'' \citep {holland1986statistics}.} Consequently, predictions
 of unit-level counterfactuals are rare and  often informal in causal inference. For
 instance, a unit's treatment effect may be approximated by the conditional average
 treatment effect (CATE) among observably similar units, under implicit assumptions
 ruling out unobserved heterogeneity.%

This lack of focus on individual counterfactuals---as well as concerns about unobserved
heterogeneity---in part explains limited takeup of standard causal inference tools and
language in subfields that rely on structural demand models. On the other hand, the
complexity of structural models makes it difficult to see how its predictions depend on
modeling assumptions. It is thus useful to understand what drives identification of
unit-level counterfactuals. We do so by interpreting structural models as explicit
restrictions on the joint distribution of $Y(\cdot)$.

We now set up notation to discuss identification formally and to introduce the assumptions
in \citet{berry2014identification}. We let $F \in \mathcal P$ denote the distribution of
the observed variables, and we let $F^* \in \mathcal P^*$ denote the distribution of $
(\br{Y (\cdot): a
\in
\mathcal A}, A, Z)$. Each $F^*$ generates a particular $F$ through $Y = Y(A)$, and thus
 $\mathcal P^*$ generates $\mathcal P$. Let $\mathcal S \subset \mathcal Y \times
\mathcal A$ denote the support of $(Y(A), A)$.\footnote{For simplicity, we assume throughout that
 all
 members of $\mathcal P$ have common support: $\P_{F}((Y, A, Z) \in E) = 0 \iff \P_{F'}(
 (Y, A, Z)
\in E) = 0$ for all $F, F'
\in \mathcal P$ and all events $E \subset \mathcal Y \times
\mathcal A \times \mathcal Z$.}

We define identification for unit-level counterfactuals: A unit-level counterfactual
$Y(a)$ is identified if we can compute it from any other $(Y(a'), a')$, with a function
$m(\cdot; F)$ that is known given the observed distribution $F$.
\begin{defn}
\label{defn:id} We say that a counterfactual $Y(a)$ is identified\footnote{This notion is
 slightly stronger than what may be natural. We require the function $m$ to link any two
 potential outcomes. An alternative definition could just require that $m$ link the
 observed outcome $ (Y, A)$ to counterfactual outcomes. When $A$ is randomly assigned,
 these two notions are identical.} at $F$ if for all $F^* \in
    \mathcal P^*$ that generates $F$, there is some function $m(a, \cdot, \cdot; F) :
    \mathcal S \to \mathcal Y$
    such that \[
        \P_{F^*}\br{Y(a) = m(a, Y(a'), a'; F)} = 1
    \]
    for all $(Y(a'), a') \in \mathcal S$. We say that all counterfactuals are identified
    under $\mathcal P^*$ if, for all $a \in \mathcal A$, $Y(a)$ is identified at all $F
    \in
    \mathcal P$.
\end{defn}
\noindent If counterfactuals are identified, then the function $m (\cdot; F)$ can be
obtained from $F$. Any counterfactual for any market can then be computed by substituting
the observed $(Y,A)$ into this function, $Y (a) = m (a, Y,A; F)$. Under
\eqref{eq:structure}, if we identify the function $\bhs$ and can compute $\xi$ from any $
 (Y(a'), a')$ with the knowledge of $F$, then we can identify counterfactuals $Y(a)$ by
 applying $Y(a) = \bhs (a,
\xi(Y(a'), a'))$.

The seminal paper by \citet{berry2014identification} shows identification in this sense
for a flexible class of structural demand models. Their result nests parametric demand
models like logit, nested logit, or BLP \citep{berry1995automobile}. To introduce their
result, we partition characteristics and prices of option $j$ into $a_j =(x_{1j}, p_j, x_
{2j})$. We write $a = (x_1, p, x_2)$. Here, $x_ {1j}
\in \R$ is a special scalar characteristic,\footnote{To nest BLP in this framework, $x_1$
can
be chosen to be any characteristic that does not have a random coefficient
\citep{berry2014identification}.} $p_j$ is price, and $x_ {2j}$ collects other
characteristics.
 In their identification argument, prices $p$ and characteristics $x_2$ do not play
 distinct roles. Let $\mathcal X$ denote the space in which $p, x_2$ take values.

{\renewcommand{\theas}{BH14-1}
\begin{as}[Linear index]
\label{as:bh_linear_index}
For some random variable $\xi \in \Xi \subset \R^J$ and some map $\bhs = \bhs_{F^*}$, the
potential outcomes $F^*$ satisfy \[\P_{F^*}\br{Y (a) =  \bhs(x_1 + \xi, p, x_2)} = 1
\quad \text{for all $a =
(x_1, p, x_2) \in \mathcal A$}.\]
\end{as}

\renewcommand{\theas}{BH14-2}
\begin{as}[Invertible demand]
\label{as:bh_invertible}

    The function $\bhs(\cdot, p, x_2)$ is invertible in its first argument: There exists
    some measurable function $\bhs^ {-1}: \mathcal Y \times \mathcal X \to \R^J$ where \[\P_{F^*}\br{x_1 + \xi = \bhs^{-1}(Y(a), p,
    x_2)} = 1 \quad \text{ for
    all $a = (x_1, p,x_2) \in \mathcal A$}. \]
    \end{as}
}

\Cref{as:bh_linear_index} is stated as Assumption 5.1 in \citet{berry2021foundations}. It
is an implication of Assumption 1 in \citet{berry2014identification}, which is a similar
index restriction on an underlying random utility model. \Cref{as:bh_invertible} is a
conclusion of Lemma 1 in \citet{berry2014identification}, justified via a ``connected
substitutes'' condition in \citet{berry2013connected}. Since the identification of demand
only relies on this implication, we impose it as a high-level assumption instead.

Combined with assumptions on instruments, \cref{as:bh_invertible,as:bh_linear_index} allow
for  identification of the function $\bhs$ by exploiting an
``index-inversion-instruments'' recipe \citep{berry2021foundations}, which returns the
following moment condition:
\[\E[\xi \mid  Z] = \E[\bhs^{-1} (Y, P, X_2)\mid  Z] -  X_1  = 0.\] The function
 $\bhs^{-1}$ is then identified through nonparametric instrumental variables
 \citep{newey2003instrumental}; see Theorem 1 in \citet{berry2014identification}. Upon
  identification of $\bhs$, the structural shock $\xi =
 \bhs^ {-1}(Y,P, X_2) - X_1$ can be computed and unit-level counterfactuals are recovered.
  The map $m$ in \cref{defn:id} can be chosen as \[Y (a) = \bhs (x_1 +
 \underbrace{\bhs^{-1}(Y (a'), p', x_2') - x_1'}_{\text{model-implied $\xi$}}, p, x_2)
 \quad a= (x_1, p, x_2),
 a'= (x_1',p', x_2'),
\] which depends on the data only through the identified structural function $\bhs$.

This identification argument is mathematically simple. It shows that parametric
restrictions in BLP, for instance, are not crucial for identification. Nevertheless, it
can be somewhat mysterious how the index and invertibility assumptions allow for
identification of $\bhs$, and what distributions over $Y(a)$ they rule out. Our central
exercise is to restate
\cref{as:bh_invertible,as:bh_linear_index} equivalently only in terms of counterfactuals
$Y(\cdot)$, without presuming a generative model of $Y (\cdot)$. This restatement
precisely clarifies the restrictions on counterfactuals made by the generative model.

\subsection{Equivalent assumptions in potential outcomes}

 Our first assumption imposes that $Y(\cdot)$ satisfy \emph{counterfactual
homogeneity}.

{\renewcommand{\theas}{CH}
\begin{as}[Counterfactual homogeneity]
\label{as:latent_homogeneity}
    For each $F^* \in \mathcal P^*$, there exists some mapping $C_{\cdot \to \cdot} = C_
    {\cdot \to \cdot , F^*}$ such that \[\P_{F^*}\br{Y
    (a') = C_{a \to a'}(Y(a))} =1 
    \text{  
    for all $a, a' \in \mathcal A$.} 
    \numberthis \label{eq:conversion}\] 
    Equivalently, for some baseline treatment $a_0
    \in
    \mathcal A$, there exists $C_0(y,a) = C_{a \to a_0}(y)$, invertible in its first
    argument, such
    that for all $a \in \mathcal A$, 
    \[
   \P_{F^*} \br{Y(a_0)  = C_0(Y(a), a)} = 1. \]
\end{as}
}
\Cref{as:latent_homogeneity} states that there is a {deterministic} mapping $C_{a \to a'}$
 that converts one counterfactual $Y(a)$ into another $Y(a')$. This mapping is common to
 \emph{all markets} in the population $F^*$. Equivalently,
 counterfactuals $Y(a')$ have zero conditional
  variance given any other counterfactual outcome $Y(a)$, over draws of markets in $F^*$:
  \[\var_ {F^*} \pr{Y (a') \mid Y
  (a) } = 0_{J\times J} \quad  \text{  for all $a, a' \in \mathcal A$.} 
  \numberthis \label{eq:zero_variance}
 \] Also equivalently, we can first convert all counterfactuals $Y(a)$ into some baseline
  outcome $Y(a_0)$, and then generating counterfactuals $Y(a')$ from $Y(a_0)$. In these
  senses, \cref{as:latent_homogeneity} restricts the heterogeneity across markets by
  restricting the intrinsic dimension of the support of potential outcomes $\br{Y(a)}_
  {a \in \mathcal A}$. The
  \emph{relationship} between $Y (a)$ and $Y(a_0)$ is kept homogeneous across all markets.
  We refer to it as \emph {counterfactual homogeneity} for this reason.

  An implication of counterfactual homogeneity is that all markets that have identical
  conditions in the data $(Y, A) = (y,a)$ must then also have identical counterfactual
  outcomes $Y(a') = C_{a \to a'}(y, a)$, for all counterfactual characteristics and
  prices $a' \in
\mathcal A$: Geometrically, if two markets have crossing demand curves $a \mapsto Y
 (a)$, then the two demand curves must be identical.  \Cref {as:latent_homogeneity} is
 also a generalization of rank invariance in standard treatment effect settings.\footnote
 {There, rank invariance
\citep{doksum1974empirical} imposes that $Y(0) = C(Y(1))$ for some monotone $C$, and if
both outcomes are continuously distributed, $C$ can be taken to be $F_{Y(0)}^{-1} \circ
F_{Y(1)}$ and invertible, for $F_ {Y(j)}$ the CDF of $Y (j)$.} Relative to rank
invariance, \cref{as:latent_homogeneity} extends to non-binary treatment and
multidimensional outcomes.

Counterfactual homogeneity rules out heterogeneity \emph{across markets}. It is not an a
priori restriction on how a particular market, say a realization $y_i (a) = Y_i(a)$ drawn
from $F^*$, may respond to counterfactual bundles $a
 \mapsto y_i(a)$. Thus, to the extent that we think of $y_i(a)$ as aggregations of
  consumers within market $i$, \cref{as:latent_homogeneity} generates flexible
  substitution patterns for any given market. What \cref{as:latent_homogeneity} does
  restrict is how consumer populations can be different across markets.

 \begin{exsq}[An economic model that violates counterfactual homogeneity] 
\label{ex:model}
 Suppose each market aggregates BLP-style preferences: \[
     Y_i(a; \xi_i, \zeta_i) = \int \frac{e^{a_j'\beta+\xi_{ij}}}{1+\sum_{k=1}^J e^
     {a_k'\beta +
     \xi_{ik}}} dG
     (\beta; \zeta_i). 
 \] However, instead of assuming that the consumer taste distributions $G
  (\cdot; \zeta_i)$ are identical across markets, perhaps certain markets $(\zeta_i = 1)$
  are more price sensitive than others $(\zeta_i = 0)$. The type of the market $\zeta_i$
  is either unobserved or insufficiently proxied by observables. Then $\zeta_i$ cannot be
  recovered from the observed data and thus unit-level counterfactuals are not
  identified, even with randomized $A$. An example with $J=1$ is shown in 
  \cref{fig:demandcurves}.
 \end{exsq}

\begin{figure}[tb]
    \centering
    \includegraphics[width=0.7\textwidth]{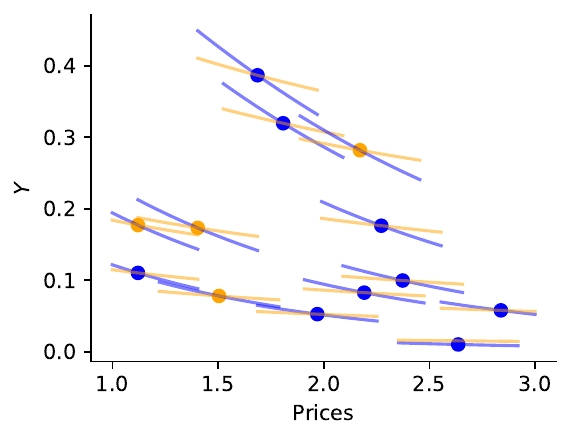}

    \begin{proof}[Notes]
        All market shares follow random coefficient logit $Y(p) = \int \Lambda
    \pr{-\alpha p + \xi} \,d G (\alpha)$, for $\Lambda(t) = 1/(1+e^{-t})$ and randomly
     assigned prices. Markets are randomly blue or orange, corresponding to $\zeta_i$ in
     \cref{ex:model}. The {\color{blue}blue} markets
     have {$\color{blue} G_ {\text{blue}}
    \sim
    \mathrm{Lognormal}(0, 0.5^2)$}. The {\color{orange}orange} markets have {$ \color
     {orange}G_ {
    \text{orange}} \sim \mathrm{Lognormal}(-0.5, 2^2)$}. For each market realization $
     (P, Y)$, we plot its own counterfactual shares at nearby price values (for {\color
     {blue}blue} markets, this is the {\color{blue}blue} curve). We also compute the
     $\xi$ value for a hypothetical market of opposite color such that its hypothetical
     demand curve passes through $(P, Y)$ (for {\color{blue}blue} markets, this is the
     {\color{orange}orange} curve). Because demand curves cross in this setting, this
     setup
     violates counterfactual homogeneity. When the colors of the markets are not
     observed, the population distribution of $ (P, Y)$ cannot perfectly distinguish
     whether a particular market is blue or orange. Since different colors imply
     different counterfactuals---including price elasticities, the counterfactuals are
     not identified. \end{proof}
    \caption{A parametric demand model with $J=1$ where counterfactual homogeneity fails
     to hold}
    \label{fig:demandcurves}
\end{figure}

The second assumption imposes some functional form restriction on the map $C_0$.
{\renewcommand{\theas}{PL}
\begin{as}[Latent partial linearity]
\label{as:latent_partial_linearity} 
For all $F^* \in \mathcal P^*$, there exists a function $h = h_{F^*}:
    \mathcal Y \times \mathcal X \to \R^J$ where, for all $a = (x_1, p, x_2) \in
    \mathcal A$, invertible in its first argument, such that \[
        \phi^{-1}\pr{C_0(y, a)} = h(y, p, x_2) - x_1,
        \numberthis \label{eq:linearity}
    \]
    for $\phi^{-1}(y) = h(y, p_0, x_{20}) - x_{10}.$
\end{as}}

\Cref{as:latent_partial_linearity}  states that, up to some invertible transformation
 $\phi$, $C_0$ is partially linear in $x_1$. This functional form restriction is
 important for identification using instrumental variables. It is also substantive,
 imposing, e.g., that $x_1$ is excluded from elasticities: the Jacobian of $Y(a)$ with
 respect to $a$ depends on $x_1$ only through $Y (a)$:
\[
    \diff{Y(a)}{x_1} = \pr{\diff{h(Y(a), p, x_2)}{y}}^{-1} \quad \diff{Y(a)}{(p, x_2)} = - \diff{Y(a)}
    {x_1} \diff{h(Y(a), p, x_2)}{(p, x_2)}. \numberthis \label{eq:derivatives}
\]

\Cref{as:latent_homogeneity,as:latent_partial_linearity} can be  combined as
 the following homogeneity assumption on some transformation of potential outcomes.

{\renewcommand{\theas}{HOM}
\begin{as}[Homogeneous effects in a transformed outcome]
\label{as:homogeneity}
    There exists some function $H(y, p, x_2) = H_{F^*}(y,p,x_2)$, invertible in $y$, such
    that the transformed potential outcome $H(a)$, for $H(a) \equiv H(Y(a), p, x_2)$,
    satisfies: \begin{enumerate}
        \item (No treatment effect in $(p, x_2)$) For all $(x_1, p_1, x_{2,1}), 
        (x_1,p_2,x_{2,2})\in \mathcal A$, \[\P_{F^*}\br{H
        (x_1,p_1,x_{2,1}) = H (x_1, p_2, x_{2,2})} = 1\]
        \item (Homogeneous linear effects in $x_1$) For all $(x_{1,1}, p, x_2),
         (x_ {1,2}, p, x_2)
        \in \mathcal A$, \[
            \P_{F^*}\br{H(x_{1,1}, p, x_2) - H(x_{1,2}, p, x_2) = x_{1,1} - x_{1,2}}=1.
        \]
    \end{enumerate}
\end{as}}

\Cref{as:homogeneity} states that for some unknown transformation of the potential outcome
 $H(a) = H(Y(a), p, x_2)$, if we treat $H(a)$ as a new potential outcome, then it admits
 no treatment effects in $(p, x_2)$ and linear treatment effects in $x_1$.\footnote
 {The slope of the $x_1$-treatment effect on $H(a)$ can be normalized through $H$.} 
\Cref{as:homogeneity} makes clear how
\cref{as:latent_homogeneity,as:latent_partial_linearity} restrict treatment effect
heterogeneity. Viewed as assumptions on some transformation of potential outcomes,
\cref{as:latent_homogeneity,as:latent_partial_linearity} are exactly constant treatment
 effects assumptions. \Cref{as:homogeneity} is weaker than standard constant treatment
 effects by not specifying which transformed outcome satisfies homogeneity---only that
 some transformation does.

Our main result is that these assumptions are equivalent to the
\citet{berry2014identification} assumptions, in the same spirit as   
\citet{vytlacil2002independence,vytlacil2006ordered}'s results for instrumental variable
models. The equivalence is easy to derive, once we link $(\bhs, \xi)$ in \cref
 {as:bh_invertible,as:bh_linear_index} to $(h,\phi, C_0)$ in
\cref{as:latent_homogeneity,as:latent_partial_linearity} and $H$ in \cref{as:homogeneity}:
\[
    \bhs = h^{-1}, \quad \xi = \phi^{-1}(Y(a_0)), \quad h(y,p,x_2) = H(y, p, x_2).
\]

\begin{restatable}{theorem}{thmmainequiv}
\label{thm:mainequiv}
The following are equivalent: 
\begin{enumerate}
     \item \cref{as:bh_invertible,as:bh_linear_index},
    \item \cref{as:latent_homogeneity,as:latent_partial_linearity},
    \item \cref{as:homogeneity}.
\end{enumerate}    
\end{restatable}

Reformulating assumptions this way retells the progress in demand models with market share
data. In the standard telling \citep{ackerberg2007econometric}, different generations of
structural demand models (e.g., vertical models, simple logit, nested logit, BLP, \citet
{berry2014identification}) all maintain random utility models of consumer behavior and
treat market shares as aggregations of consumer choices. They differ in the flexibility
of the utility model and of implied substitution patterns. In this retelling, all such
demand models instead maintain counterfactual homogeneity and latent partial linearity of
market shares. They specify different parametrized classes of $h$, which governs
model-implied substitution patterns. These two perspectives---making the random utility
model increasingly flexible versus enlarging the function class for $h$---meet at the
nonparametric model in \citet {berry2014identification}.

This reformulation also clarifies why nonparametric structural demand models are able to
identify unit-level counterfactuals. It likewise explains why these models avoid
selection assumptions on how $A$ responds to instruments. Unit-level counterfactuals are
identified because of counterfactual homogeneity. Counterfactual homogeneity likewise
means that heterogeneity in the first stage does not matter for how $A$ affects $Y$,
since different types of compliers trace out exactly the same response in $H(a)$.

\section{Discussion}
\label{sec:discussion}

\subsection{The curse of unobserved heterogeneity}

\Cref{thm:mainequiv} clarifies that structural demand models \emph{do} restrict unobserved
 heterogeneity. The need to restrict unobserved heterogeneity is not specific to these
 particular demand models either. Any model that \emph{identifies} unit-level
 counterfactuals necessarily has to impose counterfactual homogeneity: \cref
 {as:latent_homogeneity} is necessary for identification in the sense of \cref
 {defn:id}. 

\begin{restatable}[Necessity of counterfactual homogeneity]{prop}{lemmalatent}
\label{lemma:lemmalatent}
Suppose all counterfactuals are identified under $\mathcal P^*$ in the sense of
\cref{defn:id}, then \cref{as:latent_homogeneity} is satisfied.
\end{restatable}

No nonparametric model can relax counterfactual homogeneity without giving up
identification. Thus, the difference between the two cultures---structural demand
modeling and causal inference---is when each incurs this curse of unobserved
heterogeneity. Structural demand models incurs it up front, whereas causal inference
approaches implicitly incurs it when extrapolating from average treatment effects. In
either case, the fundamental problem of causal inference remains. 

Given the goal of identifying unit-level counterfactuals, \citet
{berry2014identification} impose little more than what is necessary. The functional form
assumption, \cref {as:latent_partial_linearity}, is strictly speaking not
necessary.\footnote{As a simple example, suppose we instead assumed a different,
multiplicative functional form: \[ Y_j(a_0) = \phi_j\pr { g_j(Y(a), x) \exp(-w_j) }.
\numberthis \label{eq:multiplicative}
\] When $g_j(y, x)$ can take on zero or negative values, this multiplicative formulation
 is different from \cref{as:latent_partial_linearity} because $\log(g_j(Y(a), x) \exp
 (-w_j))$ is undefined. However, we may continue to exploit $\E[g_j(Y, X) \mid W, Z] =
 c_0 \exp(W_j)$ to identify $g_j(\cdot, \cdot)$.} But it is  not relaxable without
 imposing additional assumptions, since many distinct mappings among the potential
 outcomes are observationally equivalent and satisfy counterfactual homogeneity.\footnote
 {This is clear with two treatments $(a_0, a_1)$, the set of observationally equivalent
 $C_0$ corresponds to the set of transport maps between the distributions $F_{Y
 (a_0)}$ and $F_{Y(a_1)}$. One would need some other assumption to rule out all but one
 transport map for identification. } In this sense, the assumptions in \citet
 {berry2014identification} are close to minimal for point-identification.

 Nevertheless, counterfactual homogeneity is likely misspecified: The zero-variance
 implication \eqref{eq:zero_variance} is implausible in many applications.  Economic
 models allowing for markets that differ in terms of their consumer populations,
 like \cref{ex:model}, would violate this assumption. We may have little compelling
 reason to rule out these models---other than that ruling them out makes unit-level
 counterfactuals identified. In parametric models, these restrictions are also testable
 if overidentifying moments are nonlinear in parameters \citep
 {chesher1984testing,hahn2014neglected,qian2025testing}. Omitted heterogeneity may
 explain rejection of overidentification restrictions.
 If researchers do not find counterfactual homogeneity credible, what are their options?

One option is to avoid imposing counterfactual homogeneity altogether---conceding that
point-identification of unit-level counterfactuals is too ambitious.  In some structural
contexts, researchers are willing to settle for partial identification rather than
imposing stronger assumptions \citep
{molinari2020microeconometrics,ciliberto2009market,tebaldi2023nonparametric,kalouptsidi2020partial,pakes2015moment}.\footnote{However,
the identified set for $Y(a)$ for a unit with $(Y,A, Z)$ cannot be smaller than the
conditional support $Y(a) \mid Y,A,Z$ under $F^*$. If counterfactual homogeneity does not
hold, then this conditional support can in principle be large. Thus, partial
identification alone is unlikely to be informative of individual counterfactual
outcomes.} Another alternative is to report a posterior predictive $\pi(Y(a) \mid
(Y,A,Z))$ for $\pi$ a prior on $\mathcal P^*$, where $\mathcal P^*$ allows for
counterfactual heterogeneity. Yet another option is to focus on a smaller set of
unit-level counterfactuals. If one only demands point-identification of
counterfactuals \emph{in prices}, then structural models can be relaxed to allow for
misspecification in characteristics $x_1, x_2$ \citep{andrews2023structural}.   We show
in 
\cref{sub:counterfactuals_in_prices} that such a relaxation exactly corresponds to
 allowing for counterfactual \emph{heterogeneity} in characteristics. Ongoing
 work \citep*{newpaper} additionally shows that price counterfactuals in nonparametric
 versions of these relaxations are identified by recentered instruments \citep
 {borusyak2025estimating}.

A second option treats the model as misspecified and interprets unit-level predictions as
extrapolations \citep*{andrews2025purpose}. The next subsection formalizes an
equivalence---in a context broader than demand---between extrapolation from ATEs and
making unit-level predictions under a structural model that identifies unit-level
counterfactuals. This result then allows us to separate quasi-experimental identification
of average effects from extrapolation in structural models. We can thus interpret
structural models as extrapolating from  ATEs identified through instrument variation,
thus retaining an interpretation when the model does not hold. Structural modeling serves
as an informative prior over \emph{which} ATEs to extrapolate from.

\subsection{Reinterpretation of predicted unit-level counterfactuals}
\label{sub:extrap_equiv} 

Consider a generic context where one observes outcomes, treatments, and instruments $
(Y,A,Z)$, where $Y$ need not be market shares. A common recipe for extrapolating from
ATEs is:
\begin{enumerate}[wide]
    \item Researchers specify a class $\mathcal H$ of extrapolation rules $H(Y,A)$,
    invertible in $Y$. Each function implicitly defines a potential outcome $H(a) = H(Y
    (a), a)$.
    
    \item Researchers posit that some outcome $H(A) = H(Y,A)$ is independent of
     the instrument $Z$, in the sense that certain transforms $m(H(A))$ is mean
     independent of $Z$.\footnote{Mean independence takes $m(\cdot)$ to be the identity.
     Full independence takes $m(\cdot)$ to be all bounded measurable functions. This is
     formalized in \cref{defn:extrapolate_from_ATE}} With some caveats, we may interpret
     this orthogonality as a lack of average treatment effect on the transformed outcome
     $H (a)$.\footnote{When the treatment itself is randomly assigned ($Z=A \indep Y
     (a)$), then $\E[H (A) \mid A] = \E[H(a)] = 0$ means that $a$ has no average treatment
     effects
     on $H(a)$. When only the instrument is randomly assigned, then this condition can be
     interpreted as a lack of treatment effects that are detectable through instrument
     variation.}

    \item When $(Y, A, Z) \sim F_0$, suppose the data $F_0$ identifies a unique member $H_
     {F_0}
    \in \mathcal H$ through the orthogonality restriction in (2). Researchers then
     extrapolate from the knowledge that $H_{F_0}(a)$ has no ATEs---by making a leap of
     faith that $H_{F_0}(a)$ also has no individaul treatment effects. This results in
     predictions of the form $\tilde Y (a; Y,A) = H_{F_0}^{-1}(H_{F_0}(Y, A), a)$. 
\end{enumerate}
We formalize this in \cref{defn:extrapolate_from_ATE} and call such predictions $\tilde Y$
\emph{extrapolated from averages} with respect to extrapolation rules $\mathcal H$, since
 they fundamentally extrapolate a lack of average effects to a lack of individual
 effects.

 This recipe rationalizes many informal extrapolation rules. For instance, a researcher
 who extrapolates by estimating the average treatment effect in some transformation $f
 (Y)$ (e.g. $\log Y$) implicitly takes $\mathcal H$ to be demeaned outcomes:  \[
     \mathcal H = \br{H(y,a) = f(y) - \mu(a) : \mu(\cdot)}. \numberthis 
     \label{eq:ATE_class}
 \]
 Independence with instruments pins down the average
 structural function $\mu (a) =
 \E[Y(a)]$.\footnote{The uniqueness holds, for instance, under completeness 
 \citep{newey2003instrumental}.} Predictions under this model act as if individual
 treatment effects are equal to differences in $\mu(\cdot)$: \[
     \tilde Y(a) = f^{-1}(f(Y) + \underbrace{\mu(a) - \mu(A) }_{\text{ATE in $f(Y)$}}),
     \qquad \mu (a) =
 \E[Y(a)].
 \] Predictions from quantile treatment effects similarly extrapolate by choosing
  $\mathcal H = \br{H(y,a) \in [0,1] : H(\cdot, a) \text{ is strictly
  increasing}}$ \citep{chernozhukov2005iv}.

Through this lens, \citet{berry2014identification} choose partially linear
extrapolation rules $
          \mathcal H = \br{H(y,x_1, p, x_2) = h(y, p, x_2) - x_1 : h(\cdot)} . $ We may
           thus interpret \citet{berry2014identification} extrapolating from ATEs through
           $\mathcal H$ as well. Compared to extrapolating using rules
 \eqref{eq:ATE_class}, these rules essentially trade flexibility with respect to the
  average structural function $\mu (a)=\mu (x_1, p, x_1)$ for flexibility with respect to
  $h(y, p, x_2)$.

This dual interpretation for structural models holds more broadly: Extrapolation from
averages implicitly specify structural models that identify unit-level counterfactuals,
and structural models that identify unit counterfactuals implicitly specify extrapolation
rules.

Indeed, we could instead extrapolate by positing a structural model $\mathcal P^*$ that
rationalizes the data---in which $Y = \bhs (A,
\xi)$ and unit-level counterfactuals are identified in the sense of \cref
 {defn:id}. By \cref{lemma:lemmalatent}, the model $\mathcal P^*$ must satisfy
 counterfactual homogeneity. We can thus view a member $F^* \in \mathcal P^*$ as indexed
 by a joint distribution $ (Y(a_0), A, Z) \sim Q
 \in \mathcal Q$ and a mapping $C_0(y,a) \in \mathcal C$, since any $Y(a)$ is obtained by
 $C_0^{-1}(Y(a_0), a)$. We can likewise view a structural model as specifying a class of
 $(Q, C_0) \subset \mathcal Q \times \mathcal C$ pairs.

 The following result shows that imposing such a model generates predictions equivalent to
 extrapolation using some extrapolation rules $\mathcal H$. That is, any prediction that
 extrapolates from averages can be equivalently cast under a (possibly misspecified)
 structural model. Conversely, any structural model $\mathcal P^*$ can be thought of as
 choosing extrapolation rules---with the technical caveat that $\mathcal P^*$ allows for
 combining $C_0$ with arbitrary distributions $(Y(a_0), A, Z)$ satisfying instrument
 exogeneity, which we formalize in  \cref{defn:complete}. 

\begin{restatable}{prop}{propequivextrapolate}
\label{prop:equivalence_extrapolation} Fix a class of distributions $\mathcal P$ over
 observables $(Y,A,Z)$. Extrapolation from averages  and structural models are equivalent
 in the following sense: For any $F \in \mathcal P$, let $(Y, A, Z) \sim F$ and let
 $\tilde Y_F(a; Y,A)$ be a prediction of the counterfactual $Y(a)$ for some observed unit
 $
 (Y,A)$. 

 \begin{enumerate}[wide]
     \item If $\tilde Y_F(a; Y,A)$ is extrapolated from averages with respect to  $\mathcal H$ in the sense of
      \cref{defn:extrapolate_from_ATE}, then there exists some $\mathcal P^*$ that
       identifies unit-level counterfactuals, generates $\mathcal P$, and rationalizes
       $\tilde Y$ as identified unit-level counterfactuals.

 \item Conversely, if the predictions $\tilde Y_F(a; Y,A)$ arise from some structural
  model $\mathcal P^*$ that identifies unit-level counterfactuals, rationalizes $\mathcal
  P$, and is \emph{only restricted by exogeneity and $\mathcal C$} in the sense of \cref
  {defn:complete}, then there exists some  $\mathcal H$ that rationalizes $\tilde Y$ as
  extrapolated averages in the sense of \cref{defn:extrapolate_from_ATE}.
 \end{enumerate}
\end{restatable}

\Cref{prop:equivalence_extrapolation} thus allows us to separate quasi-experimental
 identification from extrapolation in structural models. Models identifying unit-level
 counterfactuals fundamentally extrapolate from averages, and vice versa. The averages
 themselves are identified through standard quasi-experimental research designs and do
 not require restricting the joint distribution of potential outcomes. Tools and language
 from causal inference can also be helpful in assessing the internal validity of these
 average effects. 

 The value of structural models lies in providing economically motivated extrapolation
 rules $\mathcal H$, which improve on intuitively reasonable but ad hoc ones like
\eqref{eq:ATE_class}. These rules are exactly correct under the model, but can be viewed
 as approximately correct when counterfactual homogeneity approximately holds. Separating
 identification from extrapolation in this way thus clarifies what one can credibly learn
 from data and what one needs to believe to extrapolate to economically relevant
 quantities.

 \medskip

So far, we have shown that market-level counterfactuals are only identified under
counterfactual homogeneity when we only observe market-level data. Their prediction
requires extrapolation from average effects over markets in some way. This motivates
considering whether richer data can restore identification of market-level
counterfactuals without strong assumptions. 

As an idealized benchmark, since markets aggregate populations of consumers, market-level
causal effects are also average causal effects for consumers within a given market. Thus,
with exogenous treatment variation 
\emph{within} a given market at the consumer level, counterfactual outcomes for individual
 markets are identified as average treatment effects among consumers. Close to this
 idealized benchmark,
\citet{tebaldi2023nonparametric} assume that prices are exogenously assigned\footnote
 {In \citet{tebaldi2023nonparametric}, prices (insurance premiums) are deterministic
 functions of consumer age and income.  
\citet{tebaldi2023nonparametric} assume that consumers with different ages and incomes do
 not have systematically different latent preferences, given the market that they reside
 in.} for consumers participating in the California healthcare market and partially
 identify counterfactual market shares. 

The additional value of richer data similarly motivates the literature on ``micro BLP''
\citep{berry2024nonparametric,microblp,conlon2025incorporating}, where we observe market
 shares by demographic subgroups within a given market, though these subgroups are
 subjected to the same bundle of products. Do identification results these settings avoid the curse of unobserved heterogeneity? We conclude this paper by
 deriving an analogous equivalence for identification results
 with micro-data 
 \citep{berry2024nonparametric}. We find that identification with micro-data continues to
  impose counterfactual homogeneity. In fact, since these results are primarily motivated
  by relaxing dependence on instruments, they use even stronger forms of homogeneity
  instead. 

\section{Demographics-specific market shares}
\label{sec:demographics_specific_market_shares}

We observe market shares for different demographic subgroups $w \in \cW
\subset \R^J$. $a \in \mathcal A$ continues to denote treatment. Each market's potential
 outcome is a \emph{process} indexed by $w \in \cW$: $Y(a)[\cdot]: \cW
 \to [0,1]^J$. In this notation, $Y(a) [w]$ denotes market shares among
  demographics $w$ in a randomly drawn market, when prices and characteristics are
  counterfactually set to some value $a$. Analogous to \cref{defn:id}, we are interested
  in identifying the profile of market shares for a given market, at counterfactual
  values of treatment: $Y(a)[\cdot]$ for some $a\neq A$. It is useful to think of $w$ as
  analogous to a time index in panel settings. Consistent with that analogy, we use
  square brackets for $w$ to emphasize that comparisons in $w$ are not causal comparisons
  that represent counterfactual assignment of $w$.  

\citet{berry2024nonparametric} consider a structural model in which \[
     Y(a) [w] = \bhs(w, a, \xi)
\] for some function $\bhs$ and market demand shock $\xi$, under the following assumptions.\footnote{Relative to
 Assumption 1 in \citet{berry2024nonparametric}, \cref{as:index-micro} normalizes the
 index directly, following their Section 2.5. Relative to their setting, we suppressed
 other market-level interventions (their $X_t$) that may enter $\gamma$. Doing so makes
 the normalization in their Section 2.3 unnecessary, which we impose in \cref
 {as:index-micro} directly. } These assumptions nest parametric versions like \citet
 {microblp} (see 
\cref{ex:micro-blp}).

{\renewcommand{\theas}{BH24-1}
 \begin{as}[Index]
 \label{as:index-micro} $\bhs(w, a, \xi) = \sigma(\gamma(w, \xi), a)$, where $\gamma$ has
  codomain $\R^J$, and for all $j$, $
        \gamma_j(w,\xi) = g_j(w) + \xi_j.
    $
    For some fixed $w_0$, $g(w_0) = 0$ and $\frac{dg(w_0)}{dw} = I_J$. 
 \end{as}

\renewcommand{\theas}{BH24-2}
 \begin{as}[Invertible demand]
 \label{as:invertible-micro}
    For all $a \in \mathcal A$, $\sigma(\cdot, a)$ is injective on the support of
    $\gamma(w, \xi)$. 
 \end{as}

\renewcommand{\theas}{BH24-3}
 \begin{as}[Injective index]
 \label{as:invertible-index-micro}
    For all $\xi$ in its support, $\gamma(\cdot, \xi)$ is injective on $\mathcal W$. 
 \end{as}}

\cref{as:index-micro,as:invertible-micro,as:invertible-index-micro} are equivalently
 represented in counterfactual outcomes. The
 first of these equivalent assumptions is analogous to
\cref{as:latent_homogeneity}.

{\renewcommand{\theas}{CH-micro}
\begin{as}[Counterfactual homogeneity of market share profiles]
\label{as:lat_homogeneity_paths}
For some baseline treatment $a_0 \in \mathcal A$, there exists some invertible function
$C_0 (\cdot , a): \mathcal Y \to \mathcal Y$ such that for all $w \in \cW$ and all $a \in
\mathcal A$, \[ Y (x_0) [w] = C_0 (Y (a) [w], a) \quad \text{$P^*$-almost surely}.
    \]
\end{as}}

\cref{as:latent_homogeneity} posits that a deterministic, invertible function maps $Y
 (a)$ to $Y(a_0)$. Analogously, \cref{as:lat_homogeneity_paths} posits that such a
 function maps the \emph{profile} of market shares $Y (a)[\cdot]$ to $Y (a_0)
 [\cdot]$. The mapping in \cref{as:lat_homogeneity_paths} acts identically
 along the profile $w \mapsto Y(a) [w]$ and does not depend on $w$. 

{
\renewcommand{\theas}{PT}
\begin{as}[Latent individual parallel trends]
\label{as:latent-PT}
Fix baseline values $a_0, w_0$. For some invertible mapping $\phi: \mathcal Y \to
\R^J$, the profiles $w\mapsto \phi(Y(a_0)[w])$ are parallel almost surely: There exists
an invertible and differentiable function $g:
\mathcal W \to \R^J$ such that differences in $\phi(Y(a_0)[\cdot])$ are equal to
 differences in $g(w)$
    \[
        \phi(Y(a_0)[w]) - \phi(Y(a_0)[w_0]) = g(w) - g(w_0) \quad \text{ $P^*$-almost
         surely for all $w
        \in \cW$.}
    \] Redefining $\phi(\cdot)$ if necessary, we normalize $g(w_0) = 0$ and $\frac{d}{dw}g
     (w_0) = I_J$.
\end{as}}

\Cref{as:latent-PT} states that, up to some invertible transformation $\phi(\cdot)$, the
 market share profiles at some baseline treatment $w \mapsto Y(a_0)[w]$ are parallel
 almost surely. This is an individual version of the parallel trends assumption, though
 here the ``time index'' is the demographic values $w$. It imposes that trends are not
 only parallel in expectation, but are parallel almost surely.\footnote
 {In difference-in-differences applications, where $w$ is a time index, parallel trends
 is usually stated as \[
    \E[Y(x_0)[w] - Y(x_0)[w_0] \mid A=a] = g(w) 
\]
and does not depend on the realized treatment $a$. This does not require that $Y(a_0)[w] -
Y (a_0)[w_0] = g(w)$ almost surely. 

Similarly, suppose $w$ is a time-index, if potential outcomes are generated through a
two-way fixed effects model $Y_i(a)[w] = \alpha_i +
\beta[w] + f(a) + \epsilon_i[w] $, then the individual-level trends are only parallel to
 $w \mapsto \beta[w] + \epsilon_i[w]$, which depends on the path of idiosyncratic shocks
 $\epsilon_i[\cdot]$. Relative to this, \cref{as:latent-PT} effectively assumes away the
 idiosyncratic shocks $\epsilon_i[w]$. } Thus, in addition to restricting heterogeneity
 in the relationship $a\mapsto Y (a)$, \cref{as:latent-PT} restricts the heterogeneity of
 the relationship $w \mapsto Y(a_0)[w]$, at some fixed $x_0$, across markets.

\Cref{as:lat_homogeneity_paths,as:latent-PT} are further equivalent to the following
assumption by choosing $h(\cdot, a) = \phi(C_0(\cdot, a))$.
{\renewcommand{\theas}{HOM-micro}
\begin{as}[Individual parallel trends in a transformed outcome]
    \label{as:micro-data-combined}
For some fixed $w_0$, there is an invertible function $h(\cdot, x)$ such that for some
invertible and differentiable function $g$, \[
    h(Y(x)[w], x) - h(Y(x_0)[w_0], x_0) = g(w) - g(w_0) \text{ for all $x, w, x_0$},
\]
$P^*$-almost surely. Redefining $h$ if necessary, we normalize $g(w_0) = 0,
\frac{d}{dw}g(w_0) = I_J$.
\end{as}}

Analogous to \cref{as:homogeneity}, \cref{as:micro-data-combined} states that individual
parallel trends hold for transformed outcome profiles $H[w] = H(a)[w] \equiv h(Y(a)
[w], x)$, which do not depend on the treatment $x$. Thus, under \cref
{as:micro-data-combined}, there is some transformed outcome profile $H(a)[\cdot]$ that
receives no treatment effect from $a$ and has parallel sample profiles. 

We collect these equivalences in the following theorem.
\begin{restatable}{theorem}{thmequivmicro}
\label{thm:equivmicro}
    The following are equivalent:
    \begin{enumerate}
        \item \cref{as:index-micro,as:invertible-micro,as:invertible-index-micro}
        \item \cref{as:lat_homogeneity_paths,as:latent-PT}
        \item \cref{as:micro-data-combined}.
    \end{enumerate}
\end{restatable}

We conclude this section---and the paper---by explaining the identification argument in
\citet{berry2024nonparametric}, from the perspective of \cref
 {as:micro-data-combined}. This exposition highlights the strength of the homogeneity
 assumptions in delivering identification results. In short, the homogeneity structure
 embedded in \cref{as:micro-data-combined} is already powerful enough to identify $g
 (w)$ and identify $h$ up to level shifts,\footnote{That is, for some fixed baseline
 $y_0$, $h(\cdot, x) - h(y_0, x)$ can be identified.} given the distribution of observed
 data $(Y[\cdot], A) \sim F$---without any restrictions on treatment assignment. Randomly
 assigned instruments then identify the remaining unknown $h(y_0, \cdot)$. 

\begin{figure}[tb]
    \centering
    (a) Distribution of market share \emph{profiles} under some $H(Y)$, rejected by the
    data

    \includegraphics[width=0.8\textwidth]{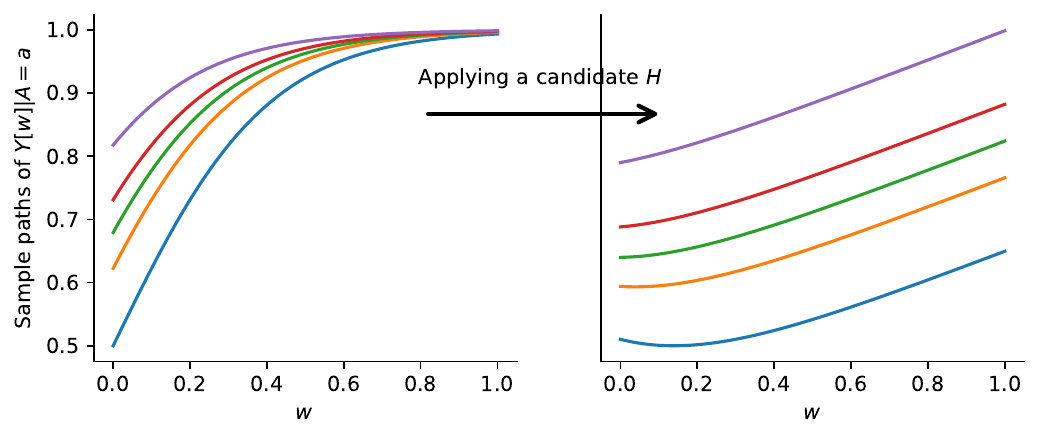}

    (b) Distribution of market share \emph{profiles} under the true $H(Y)$
    
    \includegraphics[width=0.8\textwidth]{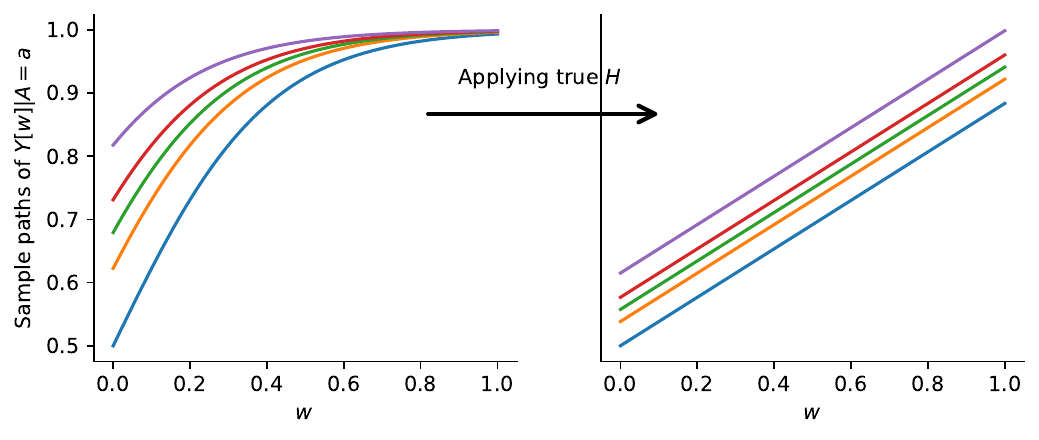}

    \caption{We show sample paths of $Y[\cdot] \mid A=a$ when $J=1$ and $H(Y[\cdot]) \mid
    A=a$ for candidate $H(\cdot)$. A candidate is rejected by the data if the sample
    paths of $H(Y[\cdot])$ are not almost surely parallel.}
    \label{fig:sample_paths}
\end{figure}

To see this, for a given value $a$,  consider the conditional distribution $Y
[\cdot] \mid A=a$. Since $A$ is not randomly assigned, this is the distribution of demand
profiles for markets that select into the product bundle $a$. On this
subpopulation,  \cref{as:micro-data-combined} states that there is some function $h
(\cdot) = h(\cdot, a)$, such that the sample paths $w
\mapsto h(Y[w])$ are almost surely parallel: \[ h \in \br{h: \P_{F}\br{h (Y[w]) - h(Y
[w_0]) =
g(w) \mid A=a} = 1 }.
\] 

Intuitively, this requirement is highly constraining: There should not be many
transformations $h$ that result in parallel profiles. In a setting with $J=1$, 
\cref{fig:sample_paths}(a) illustrates for an arbitrary candidate $H(y)$, the sample paths
 post-transformation are unlikely to be almost surely parallel, leading us to reject this
 candidate. Making the sample paths parallel seems to require getting $H$
exactly right, as in \cref{fig:sample_paths}(b). This rigidity locks in certain features
 of $h (\cdot)$. In fact, under mild smoothness and support restrictions, this rigidity
 identifies $h(\cdot)$ up to a vertical shift and $g(w)$:  Lemma 2, Lemma 3, and
 Corollary 1 in
\citet{berry2024nonparametric} show that $g(w)$ and $h (\cdot, a) - h(y_0, a)$, for some
baseline value $y_0$, are identified.

Instruments eliminate this last indeterminacy in $h(y_0, a)$.
\cref{as:micro-data-combined} implies that, for any fixed $w$, \begin{align*} h(y_0, X) &=
\overbrace{g(w) - (h(Y[w], X) - h(y_0, X))}^{\text{Identified through parallel trends}}
-
h (Y (a_0) [w_0], a_0) \\&\equiv Q (Y, w, X) - h
(Y(a_0)[w_0], a_0)
\end{align*} for an identified function $Q(Y, w, X)$. Given some instrument $Z \indep Y
 (a_0)$, we then have a moment condition that identifies $h(y_0, \cdot)$ under
 completeness \citep{newey2003instrumental}, since $
    \E[h(y_0, X) \mid Z] - \E[Q(Y, w, X) \mid Z]$ is constant in $Z$.

This intuition concurs with that in \citet{berry2024nonparametric} on the value of
micro-data and instruments. They argue that micro-data $w$ provide variation akin to
within-unit comparisons in panel data settings (p.1152). 
\cref{as:latent-PT} additionally highlights that {homogeneity}---in the sense of
\emph{individual} parallel trends---is also important, relative to standard assumptions
in panel settings. 
\Cref{as:latent-PT}, interpreted as a panel assumption, additionally imposes that the unit
 fixed effect is the only heterogeneity across units; absent the fixed effect, all units
 have the same evolution over $w$.

The equivalence \cref{thm:equivmicro}  reveals that in this model, the availability of
micro-data does not relax requirements on counterfactual homogeneity. In fact, additional
homogeneity assumptions---those with respect to $w \mapsto Y(a)[w]$---are imposed to
instead weaken requirements on instruments. Thus, whether identification results
exist---without these cross-market homogeneity assumptions and without within-market
treatment variation---remains a question for future research.

\bibliographystyle{ecca}
\bibliography{main}

\begin{thebibliography}{46}
\providecommand{\natexlab}[1]{#1}

\bibitem[{Ackerberg \textit{et~al.}(2007)Ackerberg, Benkard, Berry and
  Pakes}]{ackerberg2007econometric}
\textsc{Ackerberg, D.}, \textsc{Benkard, C.~L.}, \textsc{Berry, S.} and
  \textsc{Pakes, A.} (2007). Econometric tools for analyzing market outcomes.
  \textit{Handbook of econometrics}, \textbf{6}, 4171--4276.

\bibitem[{Andrews \textit{et~al.}(2025{\natexlab{a}})Andrews, Barahona,
  Gentzkow, Rambachan and Shapiro}]{andrews2023structural}
\textsc{Andrews, I.}, \textsc{Barahona, N.}, \textsc{Gentzkow, M.},
  \textsc{Rambachan, A.} and \textsc{Shapiro, J.~M.} (2025{\natexlab{a}}).
  Structural estimation under misspecification: Theory and implications for
  practice. \textit{The Quarterly Journal of Economics}, p. qjaf018.

\bibitem[{Andrews \textit{et~al.}(2025{\natexlab{b}})Andrews, Chen and
  Tecchio}]{andrews2025purpose}
\textsc{---}, \textsc{Chen, J.} and \textsc{Tecchio, O.} (2025{\natexlab{b}}).
  The purpose of an estimator is what it does: Misspecification, estimands, and
  over-identification. \textit{arXiv preprint arXiv:2508.13076}.

\bibitem[{Angrist \textit{et~al.}(2000)Angrist, Graddy and
  Imbens}]{angrist2000interpretation}
\textsc{Angrist, J.~D.}, \textsc{Graddy, K.} and \textsc{Imbens, G.~W.} (2000).
  The interpretation of instrumental variables estimators in simultaneous
  equations models with an application to the demand for fish. \textit{The
  Review of Economic Studies}, \textbf{67}~(3), 499--527.

\bibitem[{Angrist and Pischke(2010)}]{angrist2010credibility}
\textsc{---} and \textsc{Pischke, J.-S.} (2010). The credibility revolution in
  empirical economics: How better research design is taking the con out of
  econometrics. \textit{Journal of economic perspectives}, \textbf{24}~(2),
  3--30.

\bibitem[{Benkard and Berry(2006)}]{benkard2006nonparametric}
\textsc{Benkard, C.~L.} and \textsc{Berry, S.} (2006). On the nonparametric
  identification of nonlinear simultaneous equations models: Comment on brown
  (1983) and roehrig (1988). \textit{Econometrica}, \textbf{74}~(5),
  1429--1440.

\bibitem[{Berry \textit{et~al.}(2013)Berry, Gandhi and
  Haile}]{berry2013connected}
\textsc{Berry, S.}, \textsc{Gandhi, A.} and \textsc{Haile, P.} (2013).
  Connected substitutes and invertibility of demand. \textit{Econometrica},
  \textbf{81}~(5), 2087--2111.

\bibitem[{Berry \textit{et~al.}(2004)Berry, Levinsohn and Pakes}]{microblp}
\textsc{---}, \textsc{Levinsohn, J.} and \textsc{Pakes, A.} (2004).
  Differentiated products demand systems from a combination of micro and macro
  data: The new car market. \textit{Journal of political Economy},
  \textbf{112}~(1), 68--105.

\bibitem[{Berry and Haile(2014)}]{berry2014identification}
\textsc{Berry, S.~T.} and \textsc{Haile, P.~A.} (2014). Identification in
  differentiated products markets using market level data.
  \textit{Econometrica}, \textbf{82}~(5), 1749--1797.

\bibitem[{Berry and Haile(2021)}]{berry2021foundations}
\textsc{---} and \textsc{---} (2021). Foundations of demand estimation. In
  \textit{Handbook of industrial organization}, vol.~4, Elsevier, pp. 1--62.

\bibitem[{Berry and Haile(2024)}]{berry2024nonparametric}
\textsc{---} and \textsc{---} (2024). Nonparametric identification of
  differentiated products demand using micro data. \textit{Econometrica},
  \textbf{92}~(4), 1135--1162.

\bibitem[{Berry \textit{et~al.}(1995)Berry, Levinsohn and
  Pakes}]{berry1995automobile}
\textsc{---}, \textsc{Levinsohn, J.} and \textsc{Pakes, A.} (1995). Automobile
  prices in market equilibrium. \textit{Econometrica}, \textbf{63}~(4),
  841--890.

\bibitem[{Borusyak \textit{et~al.}(2025{\natexlab{a}})Borusyak, Bravo and
  Hull}]{borusyak2025estimating}
\textsc{Borusyak, K.}, \textsc{Bravo, M.~C.} and \textsc{Hull, P.}
  (2025{\natexlab{a}}). Estimating demand with recentered instruments.
  \textit{arXiv preprint arXiv:2504.04056}.

\bibitem[{Borusyak \textit{et~al.}(2025{\natexlab{b}})Borusyak, Chen, Hull and
  Lei}]{newpaper}
\textsc{---}, \textsc{Chen, J.}, \textsc{Hull, P.} and \textsc{Lei, L.}
  (2025{\natexlab{b}}). Supply shocks can identify pricing counterfactuals in
  nonparametric demand models.

\bibitem[{Breiman(2001)}]{breiman2001statistical}
\textsc{Breiman, L.} (2001). Statistical modeling: The two cultures (with
  comments and a rejoinder by the author). \textit{Statistical science},
  \textbf{16}~(3), 199--231.

\bibitem[{Chernozhukov and Hansen(2005)}]{chernozhukov2005iv}
\textsc{Chernozhukov, V.} and \textsc{Hansen, C.} (2005). An iv model of
  quantile treatment effects. \textit{Econometrica}, \textbf{73}~(1), 245--261.

\bibitem[{Chesher(1984)}]{chesher1984testing}
\textsc{Chesher, A.} (1984). Testing for neglected heterogeneity.
  \textit{Econometrica: Journal of the Econometric Society}, pp. 865--872.

\bibitem[{Chiappori \textit{et~al.}(2015)Chiappori, Komunjer and
  Kristensen}]{chiappori2015nonparametric}
\textsc{Chiappori, P.-A.}, \textsc{Komunjer, I.} and \textsc{Kristensen, D.}
  (2015). Nonparametric identification and estimation of transformation models.
  \textit{Journal of Econometrics}, \textbf{188}~(1), 22--39.

\bibitem[{Ciliberto and Tamer(2009)}]{ciliberto2009market}
\textsc{Ciliberto, F.} and \textsc{Tamer, E.} (2009). Market structure and
  multiple equilibria in airline markets. \textit{Econometrica},
  \textbf{77}~(6), 1791--1828.

\bibitem[{Compiani(2018)}]{compiani2018nonparametric}
\textsc{Compiani, G.} (2018). Nonparametric demand estimation in differentiated
  products markets. \textit{Available at SSRN 3134152}.

\bibitem[{Conlon and Gortmaker(2025)}]{conlon2025incorporating}
\textsc{Conlon, C.} and \textsc{Gortmaker, J.} (2025). Incorporating micro data
  into differentiated products demand estimation with pyblp. \textit{Journal of
  Econometrics}, p. 105926.

\bibitem[{Conlon and Mortimer(2021)}]{conlon2021empirical}
\textsc{---} and \textsc{Mortimer, J.~H.} (2021). Empirical properties of
  diversion ratios. \textit{The RAND Journal of Economics}, \textbf{52}~(4),
  693--726.

\bibitem[{Doksum(1974)}]{doksum1974empirical}
\textsc{Doksum, K.} (1974). Empirical probability plots and statistical
  inference for nonlinear models in the two-sample case. \textit{The annals of
  statistics}, pp. 267--277.

\bibitem[{Freyberger(2015)}]{freyberger2015asymptotic}
\textsc{Freyberger, J.} (2015). Asymptotic theory for differentiated products
  demand models with many markets. \textit{Journal of Econometrics},
  \textbf{185}~(1), 162--181.

\bibitem[{Hahn \textit{et~al.}(2014)Hahn, Newey and Smith}]{hahn2014neglected}
\textsc{Hahn, J.}, \textsc{Newey, W.~K.} and \textsc{Smith, R.~J.} (2014).
  Neglected heterogeneity in moment condition models. \textit{Journal of
  Econometrics}, \textbf{178}, 86--100.

\bibitem[{Heckman \textit{et~al.}(1997)Heckman, Smith and
  Clements}]{heckman1997making}
\textsc{Heckman, J.~J.}, \textsc{Smith, J.} and \textsc{Clements, N.} (1997).
  Making the most out of programme evaluations and social experiments:
  Accounting for heterogeneity in programme impacts. \textit{The review of
  economic studies}, \textbf{64}~(4), 487--535.

\bibitem[{Holland(1986)}]{holland1986statistics}
\textsc{Holland, P.~W.} (1986). Statistics and causal inference.
  \textit{Journal of the American statistical Association}, \textbf{81}~(396),
  945--960.

\bibitem[{Humphries \textit{et~al.}(2025)Humphries, Ouss, Stavreva, Stevenson
  and van Dijk}]{humphries2025conviction}
\textsc{Humphries, J.~E.}, \textsc{Ouss, A.}, \textsc{Stavreva, K.},
  \textsc{Stevenson, M.~T.} and \textsc{van Dijk, W.} (2025). Conviction,
  incarceration, and recidivism: Understanding the revolving door. \textit{The
  Quarterly Journal of Economics}, \textbf{140}~(4), 2907--2962.

\bibitem[{Imbens and Angrist(1994)}]{imbensangrist}
\textsc{Imbens, G.~W.} and \textsc{Angrist, J.~D.} (1994). Identification and
  estimation of local average treatment effects. \textit{Econometrica},
  \textbf{62}~(2), 467--475.

\bibitem[{Kalouptsidi \textit{et~al.}(2020)Kalouptsidi, Kitamura, Lima and
  Souza-Rodrigues}]{kalouptsidi2020partial}
\textsc{Kalouptsidi, M.}, \textsc{Kitamura, Y.}, \textsc{Lima, L.} and
  \textsc{Souza-Rodrigues, E.~A.} (2020). \textit{Partial identification and
  inference for dynamic models and counterfactuals}. Tech. rep., National
  Bureau of Economic Research.

\bibitem[{Kline and Walters(2019)}]{kline2019heckits}
\textsc{Kline, P.} and \textsc{Walters, C.~R.} (2019). On heckits, late, and
  numerical equivalence. \textit{Econometrica}, \textbf{87}~(2), 677--696.

\bibitem[{Kong \textit{et~al.}(2024)Kong, Dub{\'e} and
  Daljord}]{kong2024nonparametric}
\textsc{Kong, X.}, \textsc{Dub{\'e}, J.-P.~H.} and \textsc{Daljord, {\O}.}
  (2024). \textit{Nonparametric Estimation of Demand with Switching Costs: the
  Case of Habitual Brand Loyalty}. Tech. rep., National Bureau of Economic
  Research.

\bibitem[{Matzkin(2008)}]{matzkin2008identification}
\textsc{Matzkin, R.~L.} (2008). Identification in nonparametric simultaneous
  equations models. \textit{Econometrica}, \textbf{76}~(5), 945--978.

\bibitem[{Mogstad and Torgovitsky(2024)}]{mogstad2024instrumental}
\textsc{Mogstad, M.} and \textsc{Torgovitsky, A.} (2024). Instrumental
  variables with unobserved heterogeneity in treatment effects. In
  \textit{Handbook of Labor Economics}, vol.~5, Elsevier, pp. 1--114.

\bibitem[{Molinari(2020)}]{molinari2020microeconometrics}
\textsc{Molinari, F.} (2020). Microeconometrics with partial identification.
  \textit{Handbook of econometrics}, \textbf{7}, 355--486.

\bibitem[{Nevo and Whinston(2010)}]{nevo2010taking}
\textsc{Nevo, A.} and \textsc{Whinston, M.~D.} (2010). Taking the dogma out of
  econometrics: Structural modeling and credible inference. \textit{Journal of
  Economic Perspectives}, \textbf{24}~(2), 69--82.

\bibitem[{Newey and Powell(2003)}]{newey2003instrumental}
\textsc{Newey, W.~K.} and \textsc{Powell, J.~L.} (2003). Instrumental variable
  estimation of nonparametric models. \textit{Econometrica}, \textbf{71}~(5),
  1565--1578.

\bibitem[{Neyman(1923/1990)}]{splawa1990application}
\textsc{Neyman, J.} (1923/1990). On the application of probability theory to
  agricultural experiments. essay on principles. section 9. \textit{Statistical
  Science}, pp. 465--472.

\bibitem[{Pakes \textit{et~al.}(2015)Pakes, Porter, Ho and
  Ishii}]{pakes2015moment}
\textsc{Pakes, A.}, \textsc{Porter, J.}, \textsc{Ho, K.} and \textsc{Ishii, J.}
  (2015). Moment inequalities and their application. \textit{Econometrica},
  \textbf{83}~(1), 315--334.

\bibitem[{Qian(2025)}]{qian2025testing}
\textsc{Qian, E.} (2025). Testing for omitted heterogeneity.

\bibitem[{Rubin(1974)}]{rubin1974estimating}
\textsc{Rubin, D.~B.} (1974). Estimating causal effects of treatments in
  randomized and nonrandomized studies. \textit{Journal of educational
  Psychology}, \textbf{66}~(5), 688.

\bibitem[{Tebaldi \textit{et~al.}(2023)Tebaldi, Torgovitsky and
  Yang}]{tebaldi2023nonparametric}
\textsc{Tebaldi, P.}, \textsc{Torgovitsky, A.} and \textsc{Yang, H.} (2023).
  Nonparametric estimates of demand in the california health insurance
  exchange. \textit{Econometrica}, \textbf{91}~(1), 107--146.

\bibitem[{Torgovitsky(2019)}]{torgovitsky2019nonparametric}
\textsc{Torgovitsky, A.} (2019). Nonparametric inference on state dependence in
  unemployment. \textit{Econometrica}, \textbf{87}~(5), 1475--1505.

\bibitem[{Vuong and Xu(2017)}]{vuong2017counterfactual}
\textsc{Vuong, Q.} and \textsc{Xu, H.} (2017). Counterfactual mapping and
  individual treatment effects in nonseparable models with binary endogeneity.
  \textit{Quantitative Economics}, \textbf{8}~(2), 589--610.

\bibitem[{Vytlacil(2002)}]{vytlacil2002independence}
\textsc{Vytlacil, E.} (2002). Independence, monotonicity, and latent index
  models: An equivalence result. \textit{Econometrica}, \textbf{70}~(1),
  331--341.

\bibitem[{Vytlacil(2006)}]{vytlacil2006ordered}
\textsc{---} (2006). Ordered discrete-choice selection models and local average
  treatment effect assumptions: Equivalence, nonequivalence, and representation
  results. \textit{The Review of Economics and Statistics}, \textbf{88}~(3),
  578--581.

\end{thebibliography}

\appendix 

\counterwithin{theorem}{section}
\counterwithin{prop}{section}

\counterwithin{defn}{section}
\counterwithin{as}{section}
\counterwithin{ex}{section}

\setcounter{theorem}{0}
\setcounter{defn}{0}

\section{Proofs}

\thmmainequiv*
\begin{proof}
    (2) $\implies$ (1): \Cref{as:latent_homogeneity,as:latent_partial_linearity} implies
    that we
    can write \[
        x_1 + \phi^{-1} (Y(a_0)) = h(Y(a), p, x_2)
    \]
    for all $a = (x_1, p, x_2)$. Define $\xi = \phi^{-1} (Y(a_0))$. Thus we can write \[
        Y(a) = h^{-1}(x_1 + \xi, p, x_2).
    \]
    \Cref{as:bh_linear_index} holds by choosing $\bhs = h^{-1}$. \Cref{as:bh_invertible}
    holds by choosing $\bhs^{-1} = h$. 

    (1) $\Longleftarrow$ (2): We can write \[ Y(a) = \bhs (x_1 + \xi, p, x_2) \iff \xi =
        \bhs^{-1}(Y(a), p ,x_2) - x_1
    \]
    for all $a = (x_1, p, x_2) \in \mathcal A$. Thus, for $a_0 = (x_{10}, p_0, x_{20})$, \[
        \bhs^{-1}(Y(a), p, x_2) - x_1 = \bhs^{-1}(Y(a_0), p_0, x_{20}) - x_{10}.
    \]
    The right-hand side is some fixed invertible function of $Y(a_0)$, which we write as
    $\phi^ {-1}(Y(a_0))$. Therefore \[ Y(a_0) = \phi\pr{
            \bhs^{-1}(Y(a), p, x_2) - x_1
        }.
    \]
    We take $h = \bhs^{-1}$, which is invertible, and $C_0(y, a) = \phi( h(y, p, x_2) -
    x_1 )$.
    Since both $h$ and $\phi$ are invertible, so is $C_0$. This verifies
    \cref{as:latent_homogeneity,as:latent_partial_linearity}.

    (2) $\iff$ (3): Note that \[
        \phi^{-1}(Y(a_0)) + x_1 = h(Y(a),p, x_2)  
    \]
    implies that $h(Y(a), p, x_2)$ is constant in $p, x_2$, assumed to be invertible in
    $Y$, and
    homogeneously linear in $x_1$. On the other hand, given an invertible $H(Y(a),p,
    x_2)$,
    because it is constant in $p, x_2$ and homogeneously linear in $x_1$, we can write \[
        H(Y(a), p, x_2) - x_1 = H(Y(a_0), p_0, x_{20}) - x_{10}.
    \]
    This proves \cref{as:latent_homogeneity,as:latent_partial_linearity} by choosing
    $\phi^{-1}(y) = H(y, p_0, x_{20}) - x_{10}$, assumed to be invertible.
\end{proof}

\lemmalatent*
\begin{proof}
    By definition, for every $F^*$ that generates $F$, there exists some $m = m(\cdot,
    \cdot, \cdot; F)$ such that \[
        \P_{F^*}(Y(a) = m(a, Y(a'), a'; F)) = 1
    \]
    for all $(a, a') \in \mathcal A$. We can thus set $C_{a'\to a} = C_{a' \to a, F^*} =
    m(a, \cdot, a'; F)$.
\end{proof}

\propequivextrapolate*
\begin{proof}
    \begin{enumerate}[wide]
        \item Since $ \tilde Y$ are extrapolated from averages, then \[
            \tilde Y_F(a; Y, A) = H_{F}^{-1}(H_F(Y, A), a). 
        \]
        For a given $F \in \mathcal P$, let $Q_F$ denote the distribution of $(H_F(Y,A),
        A, Z).$ 
        Define a structural model $\mathcal P^*$: \[
         \mathcal P^* \equiv   \br{F^* : F^* \overset{d}{=}\pr{\br{H_F^{-1}(\xi, a)}_{a\in
         \mathcal A}, A,
            Z} \text{ such that } (\xi, A, Z) \sim Q_F, F \in \mathcal P}. 
        \]
        For each $F\in \mathcal P$, its corresponding $F^* \in \mathcal P^*$ rationalizes
        $F$ since $(Y(A), A, Z) \sim F$ for $(Y(\cdot), A, Z) \sim F^*$. Thus $\mathcal
        P^*$ generates $\mathcal P$. 

For any $F^* \in \mathcal P^*$, let $F \in \mathcal P$ be its corresponding observed
distribution. By assumption, consider the sole element of \[
    \mathcal H_I(F) \equiv \br{
        H \in \mathcal H : \E_F[m(H(Y,A)) \mid Z] = 0 \text{ for all $m \in \mathcal M_
        {H,F}$}
    } \numberthis \label{eq:HIFdefine}
\]
and denote it as $H_F$. 
Thus we can write $m(a, Y, a'; F) = H_F^{-1}\pr{H_F(Y(a'), a'), a}$. Thus $\mathcal P^*$
identifies unit-level counterfactuals, and $\tilde Y_F$ are exactly the predictions under
the model.

\item Conversely, suppose $\mathcal P^*$ rationalizes $\mathcal P$ and satisfies
     \cref{defn:complete}. Let $\mathcal C$ be the set of $C_0$ associated with $\mathcal
      P^*$. Define $\mathcal H = \mathcal C$ and $\mathcal M_{H,F}$ as in 
      \eqref{eq:indep_M}. Given any $F \in \mathcal P$, let $C_0, \tilde C_0$ be two
      members of the set \[
          \mathcal H_I(F) \equiv \eqref{eq:HIFdefine} = \br{H \in \mathcal H : H(Y,A)
          \indep_F Z}.
      \]
      (Note that if $F^*$ generates $F$, then $C_0$ corresponding to $F^*$ is a member of
            $\mathcal H_I(F)$, and thus it is nonempty.)
      By \cref{defn:complete}, there are $F^*, \tilde F^* \in \mathcal P^*$, where 
          $F^*$ is the distribution indexed by \[
              Q \sim \pr{C_0(Y,A), A, Z} \text{ and } C_0 \in \mathcal C
          \]
          and $\tilde F^*$ is indexed by \[
              \tilde Q \sim \pr{
                  \tilde C_0(Y,A), A, Z
              } \text{ and } \tilde C_0 \in \mathcal C.
          \] By construction, $F^*$ and $\tilde F^*$ are observationally equivalent, since
           both generate $F$. Since $\mathcal P^*$ identifies unit-level counterfactuals,
           we have that \[
               C_0(Y,A) = Y(a_0) = \tilde C_0(Y,A)
          \text{ and  }
               Y(a) = C_0^{-1}(Y(a_0), a) = \tilde C_0^{-1}(Y(a_0), a).
           \]
           Therefore $C_0 = \tilde C_0$ and $H_I(F) = \br{H_F}$ is a singleton and \[
               Y(a) = \tilde Y_F(a; Y,A) = H_F^{-1}(H_F(Y,A), a). 
           \]
    \end{enumerate}
\end{proof}

\thmequivmicro*

\begin{proof}
    (1) $\implies$ (2):  
Under the assumptions in (1), we can write \begin{align*}
Y(a_0)[w] &= \sigma(\gamma(a, \xi), a_0)\\
& = \sigma(\sigma^{-1}(\sigma(\gamma(w, \xi), a), a), a_0) 
\\
&= \sigma(\sigma^{-1}(Y(a)[w], a), a_0).
\end{align*}
We can then define $C_0(y, a) = \sigma(\sigma^{-1}(y, a), a_0)$. It is invertible because
$\sigma$ is invertible.  This verifies \cref{as:lat_homogeneity_paths}. Now, the
assumptions (1) imply
that \[
    \sigma^{-1}(Y(a_0)[w], a_0) = \gamma(w, \xi) = g(w) + \xi.
\]
Thus if we pick $\phi(y) = \sigma^{-1}(y, x_0)$, then \[
    \phi(Y(a_0)[w]) - \phi(Y(a_0)[w_0]) = g(w). 
\]
This verifies \cref{as:latent-PT}.

(2) $\implies$ (3): This is immediate when we choose $h(y, a) = \phi(C_0(\cdot, a))$.

(3) $\implies$ (1): we can write \[
    Y(a)[w] = h^{-1} \pr{
        g(w) + h(Y(a_0)[w_0], a_0), a
    }. 
\]
Thus if we define $\xi = h(Y(a_0)[w_0], a_0)$, then $Y(a)[w]$ satisfies the structural
model under (1) with the requisite invertibility conditions.
\end{proof}

\section{Extrapolation}

We formalize the extrapolation equivalence in \cref{sub:extrap_equiv}. We consider a class
of distributions $\mathcal P$ on the observed data $(Y,A, Z)$. We first formalize a
general recipe for generating unit-level predictions $\tilde Y(a; Y,A)$ for the
counterfactual at treatment $a$ for a unit with observed data $(Y,A)$. 

\begin{defn} 
\label{defn:extrapolate_from_ATE} 

Let $\mathcal H$ be a class of functions $H(y, a)$, invertible in the first argument.
Upon observing $(Y,A,Z) \sim F$, let
$\tilde Y_F(a; Y,A)$ be a prediction of the counterfactual $Y (a)$ for some unit with
realized outcomes $(Y,A)$. We say that the predictions $\tilde Y_F(a; Y,A)$ are \emph
{extrapolated from averages} with respect to extrapolation rules $\mathcal H$ and test
functions $\br{\mathcal M_{H, F}}_{H\in \mathcal H, F\in \mathcal     P}$ if 
\begin{enumerate}
    \item For every $F \in \mathcal P$, there is a unique $H(Y,A)$ that is orthogonal to
     $Z$ in terms of the test functions $\mathcal M_{H, F}$ under $F$: \[
        \br{H \in \mathcal H : \E_{F}[m(H(Y, A)) \mid Z] = 0, \text{ for all $m \in
        \mathcal M_{H,F}$}} = \br{H_ {F}}. \numberthis \label{eq:orthogonality}
    \] 
    \item For all $F \in \mathcal     F$, $\tilde Y$ is rationalized by $H_{F}(Y(a),a)
    \overset{\text{a.s.}} {=}0$: $
        \tilde Y_F(a; Y, A) = H_{F}^{-1}(H_{F}(Y, A), a). 
    $
\end{enumerate}

\end{defn}

We first fix a set of {extrapolation rules} $H
\in \mathcal H$. Each $H$ defines a transformed outcome $H(A) = H(Y,A) = H(Y(A), A)$. We
 use the data to find which $H$ corresponds to outcomes $H(a)$ with zero average
 treatment effects---in the sense that the outcome is orthogonal to the instruments
 (with respect to test functions in $\mathcal M$). The condition (1) ensures that these
 orthogonality restrictions pinpoint a unique $H \in \mathcal H$. 

 To elaborate on this last point, first, technically, we  mean ``average treatment
 effects'' as the effect of the instrument on the transformed outcome. Second,
 orthogonality with respect to the instruments is with respect to a class of test
 functions. This allows for encoding different notions of independence, such as mean
 independence or full independence: A simple choice is to have $\mathcal M_ {H, F} =
 \br{\mathrm{id}}$ for all $(H, F)$, which encodes mean independence. We can likewise
 encode full independence by choosing \[
     \mathcal M_{H, F} = \br{
         h\mapsto \pr{\one(h \in B) - \E_{F}[\one(H(Y, A) \in B)]} : \text{Measurable
         sets
         $B$}
     }.\numberthis \label{eq:indep_M}
 \]

 Upon finding a unique $H \in \mathcal H$ that yields transformed outcomes that have zero
 average effects \eqref{eq:orthogonality}, we predict outcomes by acting as if $a \mapsto
 H(Y(a), a)$ is zero almost surely, rather than just on average. If $\tilde Y$
 corresponds exactly to these predictions, then we say $\tilde Y$ extrapolates with
 respect to $\mathcal H$. 

On the flip side, consider a structural model $\mathcal P^*$ that generates $\mathcal P$.
As pointed out in the main text, if $\mathcal     P^*$ has identified unit-level
counterfactuals, then we can view $\mathcal  P^*$ as parametrized by (i) the distribution
of $(Y(a_0), A, Z)\sim Q$ and (ii) the map $C_0 \in \mathcal C$. The following definition
describes structural models in which the exogeneity of $Z$ and functional forms in the
outcome $C_0 \in \mathcal C$ are the only restrictions. 

 \begin{defn} 
\label{defn:complete} We say a structural model $\mathcal P^*$ that satisfies
 counterfactual homogeneity and generates $\mathcal P$ is \emph{only restricted
 exogeneity and $\mathcal C$} if every $Q, C_0$ in the following set indexes some member
 of $\mathcal P^*$: \[
\br{Q : Q \overset{d}{=} (C_0(Y,A), A, Z),\, (Y,A,Z) \sim F,\, C_0 \in \mathcal C,\, C_0
(Y,A) \indep_F Z} \times
\mathcal C.
 \]

 \end{defn}

 \section{Additional results}

 \subsection{Counterfactuals in prices}
 \label{sub:counterfactuals_in_prices}

To consider misspecification, we let $\mathcal  P^*_{\text{model}}$ be some set of
distributions of potential outcomes that satisfy \cref{as:homogeneity}. Suppose that the
model is misspecified for the true distributions of potential outcomes:
$F^* \not \in \mathcal  P^*_{\text{model}}$. Following \citet{andrews2023structural},  we
say that $\mathcal P^*_{\text{model}}$ is causally correctly specified for price at $F^*$
if there is some member $\tilde F^* \in \mathcal  P^*_{\text{model}} $ that generates
correct counterfactuals in prices.

\begin{defn}
    We say that $\mathcal P^*_{\text{model}}$ is causally correctly specified for price at
    $F^*$ if there is some $\tilde F^* \in \mathcal  P^*_{\text{model}} $ for which \[
        \P_{F^*}\br{Y(x_1, p', x_2) = h^{-1}(h(Y(a), p, x_2), p', x_2)} = 1 \text{ for all
        } a =
        (x_1,
        p, x_2), (x_1,
        p', x_2) \in \mathcal A
    \] where $h(y,p, x_2) = h_{\tilde F^*}(y, p, x_2)$.
\end{defn}

We show that this notion is exactly equivalent to $F^*$ satisfying a type of
counterfactual homogeneity in prices. 

\begin{as}[Counterfactual homogeneity in prices]
\label{as:cfh_price}
    For each $(x_1,x_2)$, fix some baseline price $p_0$ where $(x_1, p_0, x_2) \in
    \mathcal A$, where $p_0$ may depend on $(x_1, x_2)$. There exists some $C_0 (Y (a), p,
     x_2)$, invertible in its first argument, such that \[
        \P_{F^*} \br{Y(x_1, p_0, x_2)  = C_0(Y(a), p, x_2; p_0)} = 1
    \]
    for all $a = (x_1, p, x_2) \in \mathcal A$. 
\end{as}

\Cref{as:cfh_price} is the analogue of \cref{as:latent_homogeneity}, except that we are
 only transporting potential outcomes along prices, from $(x_1, p, x_2)$ to $(x_1, p_0,
 x_2)$. \Cref {as:cfh_price} implies that counterfactuals in prices are homogeneous, in
 the sense that \[
    \var_{F^*} (Y(x_1, p', x_2) \mid Y(x_1, p, x_2)) = 0
\] for all $(x_1, p', x_2), (x_1, p, x_2) \in \mathcal A$. Since \cref{as:cfh_price} makes
 no restrictions across $(x_1, x_2)$ values, the distribution of $Y(x_1', p, x_2') \mid Y
 (x_1, p, x_2)$ is not similarly restricted. Thus, relative to 
 \cref{as:latent_homogeneity}, \cref{as:cfh_price} allows for counterfactual heterogeneity
  in characteristics. \Cref{as:cfh_price} also additionally imposes that the map $C_0$
  does not depend on $x_1$ except through $p_0$, which is a side effect of imposing 
  \cref{as:latent_partial_linearity} in $\mathcal P^*_{\text{model}}$.

\begin{prop}
    $\mathcal P^*_{\text{model}}$ is causally correctly specified for price at
    $F^*$ if and only if $F^*$ satisfies \cref{as:cfh_price} and $\mathcal P^*_{
    \text{model}}$ is sufficiently rich to include $C_0$: \[
        C_0(y, p, x_2; p_0) = h_{\tilde F^*}^{-1}(h_{\tilde F^*}(y, p, x_2), x_1, p_0,
        x_2)
    \]
    for some $\tilde F^* \in \mathcal P_{\text{model}}^*$ for all $a = (x_1, p, x_2) \in
    \mathcal A$.
\end{prop}

\begin{proof}

``If'' direction: Assume $\mathcal P^*_{\text{model}}$ is causally correctly specified for
price. Then for some $h=h_{\tilde F^*}$ and $\tilde F^* \in \mathcal P^*_{\text
{model}}$,\[
    \P_{F^*}\br{
        Y(x_1, p_0, x_2) = h^{-1}(h(Y(a), p, x_2), p_0, x_2)
    } = 1
\]
We can thus take \[
    C_0(y, p, x_2; p_0) = h^{-1}(h(y, p, x_2), p_0, x_2).
\]

``Only if'' direction: On the other hand, if \[
    C_0(y,p,x_2; p_0) = h^{-1}(h(y, p, x_2), p_0, x_2)
\]
for some $h = h_{\tilde F^*}$. Then for $a = (x_1, p, x_2)$, \[
    Y(x_1, p', x_2) = C_0^{-1} (C_0 (Y(a), p, x_2; p_0), p', x_2 ; p_0) = h^
    {-1}\pr{h(Y(a), p, x_2), p', x_2}.
\]
This proves that $\mathcal P^*_{\text{model}}$ is causally correctly specified. 
\end{proof}

 \subsection{Parametric example for micro-data}

  \begin{exsq}[Micro BLP]
  \label{ex:micro-blp}
A parametric mixed logit model posits 
 \citep{microblp}, \[Y_j (a)
 [w] = \bhs_j(w, a, \xi) \equiv \int
  \frac{e^{a_j'\beta + \xi_j}}{1 + \sum_{k=1}^J e^{a_k'\beta + \xi_j}} d F(\beta \mid w).
  \numberthis 
  \label{eq:mixed_logit_micro}
 \] for some parametrized distribution $F(\beta \mid w)$. One popular choice 
 \citep{conlon2025incorporating} sets $F(\beta \mid w) \sim \Norm(\Pi w, \Sigma) $,
  parametrized by coefficient matrix $\Pi$ and variance-covariance matrix $\Sigma$. 
    
  \Cref{as:invertible-micro,as:index-micro,as:invertible-index-micro}  generalize versions
   of \eqref{eq:mixed_logit_micro} \citep{berry2024nonparametric}. Suppose $a_j$ includes
   price, and prices do not have demographic-varying coefficients. There is also a
   product fixed effect $\nu_j$ that does have demographic-varying coefficients.\footnote{These
 assumptions are restrictive. However, one could accommodate other characteristics $\tilde
 x$ and price coefficients that vary by additional demographics $\tilde w$. In such a
 model, this analysis studies the dependence on $w$ in $w \mapsto Y(a, \tilde x)[w, \tilde
 w]$ in the subpopulation that holds $\tilde x$ fixed. See Section 3 and footnote 24 in
 \citet{berry2024nonparametric}.} Moreover, suppose the random coefficient distribution is
 Gaussian. Then,
 \eqref{eq:mixed_logit_micro} is---using $\tilde \xi$ to denote the demand shock---\[ Y
 (a) [w] = \int \frac{e^ {\pi_j'w + \nu_j - \alpha a_j
 + \tilde\xi_j }}{1 + \sum_{k=1}^J e^{ \pi_k'w + \nu_k - \alpha a_k + \tilde \xi_k }}
 p_{\Norm(0,\Sigma)}(\nu) \,d\nu \equiv \sigma\pr{\Pi w + \tilde \xi, x}\, \Pi \equiv 
 \colvecb{3}
 {\pi_1'}{\vdots}{\pi_J'} \in \R^{J\times J}.
 \]
 $\sigma(\cdot, a)$ is injective because mixed logit market shares are invertible. If the
 matrix $\Pi$ is full rank, then we can reparametrize \[ g(w)= w - w_0 \quad \xi =
 \Pi^{-1} \tilde \xi,
 \]
 and absorb $\Pi, w_0$ into the demand function $\sigma$.
\end{exsq}

\end{document}